\documentclass[letterpaper, 10pt, conference]{ieeeconf}
\IEEEoverridecommandlockouts     
\usepackage{graphicx}
\usepackage{subfig}
\usepackage{multirow}
\usepackage{array}
\DeclareGraphicsExtensions{.eps,.pdf,.png,.jpg,.mps}
\usepackage{minibox}

\usepackage{amsmath} % assumes amsmath package installed
\usepackage{amssymb}  % assu
\usepackage{algorithm}
\usepackage{algpseudocode}
\usepackage{cite}
\usepackage{hyperref}
\usepackage{breakurl}
\usepackage{footnote}
\usepackage{mathtools}
\usepackage{balance}
\newtheorem{theorem}{Theorem}

\usepackage{todonotes}

\title{Feedback Scheduling for Energy-Efficient Real-Time Homogeneous Multiprocessor Systems} 

\author{Mason Thammawichai and Eric C.\ Kerrigan% <-this % stops a space
\thanks{Mason Thammawichai is with the Department of Aeronautics, Imperial College London, London SW7 2AZ, UK.
        {\tt\footnotesize m.thammawichai12@imperial.ac.uk}}%
\thanks{Eric C.~Kerrigan is with the Department of Electrical \& Electronic Engineering  and the Department of Aeronautics, Imperial College London, London SW7 2AZ, UK.
        {\tt\footnotesize e.kerrigan@imperial.ac.uk}}%
}

\begin{document}
\maketitle

\begin{abstract}
Real-time scheduling algorithms proposed in the literature are often based on worst-case estimates of task parameters. The performance of an open-loop scheme can be degraded significantly if there are uncertainties in task parameters, such as the execution times of the tasks. Therefore, to cope with such a situation, a closed-loop scheme, where feedback is exploited to adjust the system parameters, can be applied. We propose an optimal control framework that takes advantage of feeding back information of finished tasks to solve a real-time multiprocessor scheduling problem with uncertainty in task execution times, with the objective of minimizing the total energy consumption. Specifically, we propose a linear programming based algorithm to solve a workload partitioning problem and adopt McNaughton's wrap around algorithm to find the task execution order. The simulation results illustrate that our feedback scheduling algorithm can save energy by as much as 40\% compared to an open-loop method for two  processor models, i.e.\ a PowerPC 405LP and an XScale processor.     
\end{abstract}

\section{Introduction}
Computing devices, such as server farms, data centers, portable devices and desktops, will consume more than 14\% of global electricity consumption by 2020~\cite{vereecken2010}. As the performance and speed of processors increase, the  challenges in designing these future high-performance computing systems are processor power consumption and heat dissipation. Moreover, these systems may need to operate under tight energy requirements while guaranteeing a quality of service.

As specified by the Advanced Configuration and Power Interface (ACPI)~\cite{Acpi2010}, which is an open industry standard for device configuration as well as power and thermal management, the power usage of a device can be controlled by various methods. For example, by controlling the time in the idling power states, changing the operating frequency in the performance states or by putting a CPU to sleep in throttling states when the CPU temperature is critically high. 

Dynamic Voltage and Frequency Scaling (DVFS) techniques have been widely used as an energy management scheme in modern computing systems. Typically, a processor running at a higher clock frequency consumes more energy than a processor running at a lower clock frequency. Hence, DVFS techniques aim to reduce the power/energy consumption by dynamically controlling the CPU operating frequency/voltage to match the workload. Since timeliness is an important aspect for real-time systems, the main consideration in applying DVFS is to ensure that deadline constraints are not violated.

Though a lot of work has been proposed to solve real-time scheduling problems, most of them are based on the assumption that the computational task parameters, e.g.~the task's execution time, period and deadline, do not change. In other words, they are open-loop controllers. Though an open-loop scheduler can provide good performance in a predictive environment, the performance can be degraded in an unpredictable environment, where there are uncertainties in task parameters. Specifically, the actual execution time of the task can vary by as much as 87\% of measured worst-case execution times~\cite{Wegener2001}. Since it is often the case that the task parameters are based on the worst-case, it follows that the system workload is overestimated, resulting in higher energy consumption due to non-optimal solutions. Therefore, in this work, we aim to apply feedback methods from control theory to address a scheduling problem subjected to time-varying workload uncertainty.

%\subsection{Related Work}
Only a few works have adopted feedback methods from control theory to cope with a dynamic environment for real-time scheduling. For example, \cite{soria2005} proposed an energy-aware feedback scheduling architecture for soft real-time tasks for a uniprocessor. A proportional controller adjusts the workload utilization\footnote{The utilization of the task is defined as the ratio between the task execution time and its deadline. For this work, we will use the term `density' rather than utilization; in the literature, utilization is often used for a special case of a periodic taskset, i.e.\ when the task deadline is equal to its period.} through a variable voltage optimization unit. Specifically, the controlled variable is the energy savings ratio and the manipulated variable is the worst-case utilization. 

Similarly,~\cite{deepak2002} proposed a feedback method for estimating  execution times to improve the system performance, i.e.\ the number of tasks that meet  deadlines and the number of tasks that are admitted to the system. That is, the estimated execution time is calculated at each decision time interval based on the  deadline miss and  rejection ratios. 

In~\cite{zhu2005}, a feedback method was developed for a uniprocessor hard real-time scheduling problem with DVFS to cope with varying execution time tasksets. In the same manner, the actual execution time of the task is fed back to a PID controller to adjust the estimated execution time of the task, as well as the execution frequency. 

A two-level power optimization control on a multi-core real-time systems was proposed in~\cite{xing2011}. At the core-level, the utilization of each CPU is monitored and a DVFS scheme is implemented in response to uncertainties in task execution times in order to obtain a desired utilization. To further reduce power consumption, task reassignment and idle core shutdown schemes were employed at the processor level. 

All of the work in this area only consider feedback of real-time scheduling as regulation problems. However, our work will consider real-time multiprocessor scheduling as a constrained optimal control problem~\cite{mason2015}, which can be combined with a feedback scheme to handle uncertainties in an unpredictable scheduling environment, as is done in model predictive control~\cite{Mayne20142967}. Our proposed scheme would also be known as a slack reclamation scheme in the real-time scheduling literature, in which the slack time due to early completion of a task is exploited to reduce energy consumption by decreasing the operating speed of the remaining tasks in the system~\cite{zhu2001,jian2006}.     

%\subsection{Paper Contributions}
The main contributions of this paper are:
\begin{itemize}
\item A feedback and optimal control framework is proposed to solve a real-time scheduling problem with uncertainty in task execution times on a homogeneous multiprocessor system with DVFS capabilities. 
\item A convex optimization formulation is proposed to solve a workload partitioning problem.
\item The first energy-optimal scheduling algorithm to solve multiprocessor scheduling with aperiodic tasksets.
\item Though we introduce the problem with discrete frequency level systems, the framework can  be applied to continuous frequency multiprocessor systems by simply replacing the workload partitioning algorithm by the nonlinear programming formulation proposed in~\cite{mason2015}.  
\end{itemize}

%\subsection{Outline of Paper}
Details of the system model is given in Section~\ref{sec:model}. The feedback scheduling framework is presented in Section~\ref{sec:FS}. That is, Section~\ref{sec:FSc} describes scheduling as an optimal control problem, Section~\ref{sec:FSlp} presents an LP formulation to solve the problem and the overall feedback scheduling architecture is provided in Section~\ref{sec:feedback}. Simulation results to demonstrate the performance of our feedback algorithm are given in Section~\ref{sec:sim}. Lastly, we summarise the results and discuss future work in Section~\ref{sec:con}  

\section{Task and Processor Models}\label{sec:model}
%\subsection{Task and Processor Models}
A task $T_i$ is assumed to be aperiodic and defined as a triple $T_i:=(b_i,c_i,d_i)$, where $b_i$ is the task arrival time, $c_i$ is the estimated number of CPU cycles to complete the task and $d_i$ is the task relative deadline, i.e.\ a task $T_i$ arriving at time $b_i$ has a deadline at time $b_i+d_i$. The estimated minimum execution time $\underline{x}_i$ is the estimated execution time of the task $T_i$ when executed at the maximum clock frequency $f_{max}$, i.e~$\underline{x}_i:=c_i/f_{max}$. The minimum task density $\delta_i$ is defined as the ratio between the task minimum execution time and deadline, i.e.\ $\delta_i:=\underline{x}_i/d_i.$ The actual minimum execution time of the task $\underline{y}_i$ is the actual execution time when the task is executed at clock frequency $f_{max}$, i.e.\ $\underline{y}_i:=\gamma_i\underline{x}_i$, where $0<\gamma_i\leq 1$ is the estimation factor. Note that the actual execution time of the task is not known until the task has finished. We will assume that the tasks can be preempted at any time, i.e.\ the execution of the task on a processor can be suspended in order to start executing another task. Moreover, task migration is allowed, i.e.\  execution is allowed to be suspended on one processor and able to be continued on another processor. There is no delay with task preemption or migration, since we assume that the delay is added to the estimated task execution times or that the delay is negligible. Lastly, it will also be assumed that tasks do not have any resource or precedence constraints, i.e.\ the task is ready to start upon its arrival time.

%\subsection{Processor Model}
For this work, we assume a practical processor model, i.e.\ a processor has a finite set of operating frequency levels. Additionally, the processors are homogeneous, that is, having the same set of operating frequencies and power consumptions. The processor voltage/frequency can be adjusted individually using a DVFS technique.

%\subsection{Energy Consumption Model}
The energy consumed during the time interval $[t_1,t_2]$ is 
\begin{equation}
E(t_1,t_2):=\int_{t_1}^{t_2}{P(s(t))dt},
\end{equation}
where $P(s(t))$ is the instantaneous power consumption of executing a task at an execution speed $s(t)$, defined as the ratio between the operating frequency $f(t)$ to $f_{max}$, i.e.\ $s(t):=f(t)/f_{max}$. The  energy   consumed by executing and completing task $T_i$  at a constant speed~$s_i$ is the summation of the energy in the active and idle modes, hence
%\begin{subequations} 
%\begin{align}
$E(t_1,t_2)
%&=\frac{c_i}{f_{max}} \frac{(P_{active}(s_i)-P_{idle})}{s_i} + P_{idle}d_i  \\
=\underline{x}_i(P_{active}(s_i)-P_{idle})/s_i+ P_{idle}(t_2-t_1)$,
%\end{align}
%end{subequations} 
where $P_{active}(s_i)$ %:=P_{d}(s_i)+P_{s}$ 
is the  power while active and $P_{idle}$ is the idle power. Note that $P_{idle}(t_2-t_1)$ is not a function of  speed, hence can be omitted when minimizing  energy.

\section{Feedback Scheduling}\label{sec:FS}
\subsection{Continuous-time Optimal Control Problem}\label{sec:FSc}
This section recalls an optimal control formulation of a multiprocessor scheduling problem with the objective to minimize the total energy consumption~\cite{mason2015}. The problem statement is: Given $m$ homogeneous processors and $n$ real-time tasks, determine a schedule for all tasks within a time interval $[t_1,t_2]$ that solves the following infinite-dimensional continuous-time optimal control problem:
\begin{subequations} \label{cprob}
\begin{align}
& \underset{\begin{subarray}{c}
           x(\cdot),a(\cdot) \end{subarray}}{\text{minimize}} \mathrlap{\quad\int_{t_1}^{t_2}\sum_{i,k,q} a_{ik}^q(t)(P(s^q)-P_{idle})dt}\\
	& \text{subject to } \nonumber\\
	     &\quad x_{i}(b_i) = \underline{x}_i,&&\forall i \label{he1}\\ 
			 &\quad x_{i}(t) = 0,&&\forall i,t\notin [b_i,b_i+d_i) \label{he2}\\ 
       &\quad \dot{x}_{i}(t) = -\sum_{k,q} s^qa_{ik}^q(t),&&\forall i,t,\quad \text{a.e.} \label{he3}\\
			 &\quad \sum_{k,q} a_{ik}^q(t) \leq 1,&&\forall i,t \label{he4}\\
			 &\quad \sum_{i,q} a_{ik}^q(t) \leq 1,&&\forall k,t  \label{he5}\\
			 &\quad a_{ik}^q(t) \in \{0,1\},&&\forall i,k,q,t \label{he7} 
\end{align}
\end{subequations}
% where $P_{idle}$ is the power consumption when the processor is in the idle state, 
where $x_i(t)$ is the remaining estimated minimum execution time of task $T_i$, $a_{ik}^q=1$ denotes that processor $k$ executes task $T_i$ at speed level $q\in Q:=\{1,\ldots,\ell\}$ at time $t$, where $s^q$ is the corresponding speed and $\ell$ is the total number of non-idle speed levels of a processor. If  $I:=\{1,\ldots,n\},K:=\{1,\ldots,m\}$ then $\forall i,\forall k,\forall q,\forall t$ will be used as short-hand for $\forall i\in I,\forall k\in K,\forall q\in Q,\forall t\in [t_1,t_2]$, respectively. 

The objective is to minimize   energy consumption. The  estimated execution time and deadline constraints are specified in~(\ref{he1}) and~(\ref{he2}), respectively. The scheduling dynamic~(\ref{he3}) is represented by a flow model (an integrator) with the state $x$ and control input $a:=(a^1,\ldots,a^\ell)$. Constraints (\ref{he4}) and (\ref{he5}), respectively, ensure that at all times a task is not assigned to at most one non-idle processor and vice versa. Constraint~(\ref{he7}) indicates   assignment variables are binary.  

\subsection{Discrete-time Optimal Control Problem as an LP}\label{sec:FSlp}
It was shown in~\cite{mason2015} that for a practical system, where each processor has a discrete set of operating frequencies, the problem~(\ref{cprob}) can be simplified into two steps: (i) solving a workload partitioning problem using a linear programming (LP) formulation and (ii) given a solution to the workload partitioning problem, solve a task ordering problem using McNaughton's wrap around algorithm~\cite{mc1959}.

\subsubsection{Workload Partitioning}
By relaxing the constraint~(\ref{he7}) so that the value of $a$ is interpreted as the fraction of the task execution time during each discretization time interval, the workload partitioning problem can be formulated as a finite-dimensional LP (annotated as LP-DVFS). For this purpose, let $w_i^q[\mu]\in[0,1]$ denote the fraction of the  interval $[\tau_\mu,\tau_{\mu+1}]$ during which  task $T_i$ is to be executed at speed level $q$.

%\todo{This paragraph still needs work. Difficult to understand for a first-time reader}
Let $T:=\{T_i\mid i\in I\}$ denote a taskset composed of all active tasks within $[t_1,t_2]$. Let $\{\tau_0,\tau_1,\ldots,\tau_N\}$ be the set of  times corresponding to the distinct task arrival times and deadlines within the time interval $[t_1,t_2]$, where $t_1=\tau_0<\tau_1<\ldots<\tau_N=t_2$. Let $U:=\{0,1,\ldots,N-1\}$ and define a task arrival time mapping $\Phi_b:T\rightarrow U$ by $\Phi_b(T_i):=\mu$ such that $\tau_\mu=b_i,~\forall T_i\in T$, a task deadline mapping $\Phi_d:T\rightarrow U\cup\{N\}$ by $\Phi_d(T_i):=\mu$ such that $\tau_\mu=b_i+d_i,~\forall T_i\in T$ and $\mathcal{U}_i:=\{\mu\in U\mid\Phi_b(T_i)\leq\mu<\Phi_d(T_i)\},~\forall i\in I$.

The workload partitioning statement is: Given $m$ homogeneous processors and a taskset $T$ with $n$ tasks, determine the fraction of task execution times  within each  time interval  that solves the following discrete-time optimal control problem:
\begin{subequations} \label{probLP}
\begin{align}
& \underset{\begin{subarray}{c}
       \xi[\cdot],w[\cdot]\end{subarray}}
			 {\text{minimize}} \quad\mathrlap{\sum_{\mu,i,q}(\tau_{\mu+1}-\tau_\mu)  w_{i}^q[\mu](P(s^q)-P_{idle})}\\
	& \text{subject to } \nonumber\\
	     &\xi_{i}[\Phi_b(T_{i})]=\underline{x}_i & &\forall i \label{lp1}\\ 
			 &\xi_{i}[\mu]= 0,  & &\forall i,\mu\notin \mathcal{U}_i \label{lp2}\\ 
       &\xi_{i}[\mu+1]=\xi_{i}[\mu]\nonumber\\&\qquad - (\tau_{\mu+1}-\tau_\mu)\sum_{q}s^qw_{i}^q[\mu], & &\forall i,\mu\in \mathcal{U}_i \label{lp3}\\   	
			 &\sum_{q}w_{i}^q[\mu] \leq 1,&&\forall i,\mu\in U \label{lp4}\\ 
			 &\sum_{i,q} w_{i}^q[\mu]\leq m,&&\forall \mu\in U\label{lp5}\\
			 &0 \leq w_{i}^q[\mu] \leq 1,&&\forall i,q,\mu\in U\label{lp6}
\end{align}
\end{subequations}
where the state $\xi_i[\mu]$ is the estimated minimum  execution time  of task~$T_i$  and $w_i^q[\mu]$ can be interpreted as the value of a control input at time instant~$\tau_\mu$.

The constraints on the dynamics (\ref{lp1})--(\ref{lp3}) correspond to \eqref{he1}--\eqref{he3}. Constraint~(\ref{lp4}) assures that a task will not be assigned to more than one processor at a time. 
Constraint~(\ref{lp5}) guarantees that the total workload during each time interval will not exceed the system capacity. Lastly,~(\ref{lp6}) provides the appropriate lower and upper bounds on $w_i^q[\mu]$. 

The  functions $\xi:U\cup\{N\}\rightarrow \mathbb{R}^n$ and $w:=(w^1,\ldots,w^\ell):U\rightarrow \mathbb{R}^{m \times \ell}$ map finite sets to the Euclidean space, hence it follows that~\eqref{probLP} is equivalent to a finite-dimensional LP with a tractable number of decision variables and constraints. Note that many of the components of the solution are always zero and that the LP is highly structured with sparse matrices and vectors. These facts can be exploited to develop  efficient tailor-made solvers, as in the literature on model predictive control~\cite{Mayne20142967}.

Note that  $w_i^q[\mu]$ does not have a subscript $k$ to indicate processor assignment, which is done during  task ordering.

\subsubsection{Task Ordering}
Given a solution to~(\ref{probLP}), we can find an execution order for all tasks within each time interval such that no task is executed on more than one non-idle processor at each time instant. This can be done  using McNaughton's wrap around algorithm~\cite{mc1959},  which is detailed in Algorithm~\ref{mc} for the problem considered here\footnote{Note that this version of McNaughton's algorithm is to simplify the presentation in this paper --- there could be better ways to order tasks and modes to minimise preemptions, migrations, etc.}. 
\begin{algorithm}[!t]
\caption{McNaughton's wrap around algorithm~\cite{mc1959}}\label{mc}
\begin{algorithmic}[1]
\State \bf{INPUT} $\{w_i^q[\mu]\in[0,1]\mid i\in I,q\in Q\}$
\State $\sigma_{ik}^q[\mu]\gets 0,\eta_{ik}^q[\mu]\gets 0, \forall i,k,q$
\State $k \gets 1$
\For{$i=1,\ldots,n$}
\For{$q=1,\ldots,\ell$}
\If{$i = 1$}
	\State $\eta^q_{11}[\mu]\gets w_1^q[\mu]$
\Else
		\If {$\eta^q_{(i-1)k}[\mu]+w_i^q[\mu] \leq k$}
			\State $\sigma^q_{ik}[\mu]\gets \eta^q_{(i-1)k}[\mu]$
			\State $\eta^q_{ik}[\mu]\gets \sigma^q_{ik}[\mu]+w_i^q[\mu]$
		\Else
			\State $\sigma^q_{ik}[\mu]\gets \eta^q_{(i-1)k}[\mu]$
			\State $\eta^q_{ik}[\mu]\gets 1$
			\State $\eta^q_{i(k+1)}[\mu]\gets w_i^q[\mu]-(\eta^q_{ik}[\mu]-\sigma^q_{ik}[\mu])$
			\State $k\gets k+1$
		\EndIf
\EndIf
\EndFor
\EndFor
\State \bf{RETURN} $\{(\sigma^q_{ik}[\mu],\eta^q_{ik}[\mu])\in[0,1]\times [0,1]\mid  i\in I,k \in K,q\in Q\}$
\end{algorithmic}
\end{algorithm}

The algorithm proceeds as follows for a given interval $[\tau_\mu,\tau_{\mu+1}]$.
The   fractions~$\{w_i^q[\mu]\in[0,1]\mid i\in I,q\in Q\}$ %corresponding to non-idle modes  
care aligned  in an  order by task, with modes  grouped together by task, along the real number line starting at zero. The line is split at each natural number 1, 2, etc., with each chunk  assigned to one processor. Tasks that have been split (called migrating tasks) are assigned to two different processors at non-overlapping time intervals. The algorithm returns  $\{(\sigma^q_{ik}[\mu],\eta^q_{ik}[\mu])\in[0,1]^2\mid  i\in I,k \in K,q\in Q\}$, which is used to define the start and end times of tasks on processors during an interval. 
Processor~$k$  starts to work on task $T_i$ at mode $q$ at time $\tau_\mu+\sigma_{ik}^q[\mu](\tau_{\mu+1}-\tau_\mu)$ and ends at time $\tau_\mu+\eta_{ik}^q[\mu](\tau_{\mu+1}-\tau_\mu)$. 
%The processor is idle if a task is not assigned to it. 
%For a given task, each non-idle speed level $q\neq 0$ has to be selected for a total of $w_i^q[\mu](\tau_{\mu+1}-\tau_\mu)$ seconds during the interval. To ensure that this constraint is met for migrating tasks,  appropriate care has to be taken  in case it is necessary to execute a task at the same speed level on different processors.

Consider the  taskset composed of four tasks are to be scheduled on two homogeneous processors with two non-idle modes. Suppose execution  fractions in a time  interval is as shown in Table~\ref{tab:ex1}. 
\begin{table}[t]
\caption{Execution workload partition example} \label{tab:ex1}
\centering
\begin{tabular}{|c|c|c|c|c|c|}\hline
Task  & $w_i^1[\mu]$ & $w_i^2[\mu]$ & Task  & $w_i^1[\mu]$ & $w_i^2[\mu]$  \\\hline
$T_1$  & 0.1 & 0.2 & $T_3$ & 0.2 & 0.4  \\\hline
$T_2$  & 0 & 0.5 & $T_4$  & 0.4 & 0  \\\hline
\end{tabular}
\end{table}     
Figure~\ref{fig:ex1} illustrates a feasible schedule of the taskset using McNaughton's wrap around algorithm.
\begin{figure}[t]
\centering
\subfloat[Tasks are aligned along the real number line.]{\includegraphics[width=\columnwidth]{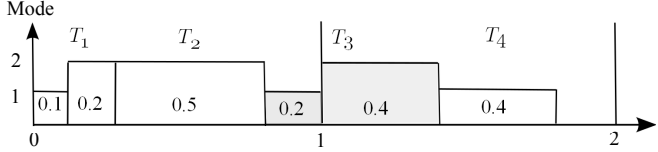}}\newline
\subfloat[Each chunk of length 1 is assigned to a processor.]{\includegraphics[width=0.8\columnwidth]{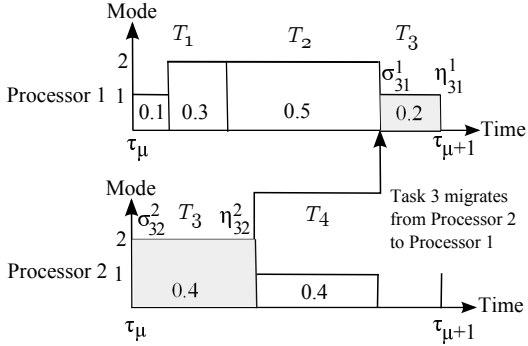}}
\caption{Feasible schedule at time interval $[\tau_\mu,\tau_{\mu+1}]$ obtained by McNaughton's wrap around algorithm, where the number in each box is $w_i^q[\mu]$.}
\label{fig:ex1}
\end{figure}

We are now in a position to state the following.
\begin{theorem} A solution  to~(\ref{cprob}) can be used to construct a solution to~(\ref{probLP}). Furthermore, a solution  to~(\ref{cprob}) can be constructed from a solution to~(\ref{probLP}) and the output from Algorithm~\ref{mc}.\end{theorem}
\begin{proof}
%We  only need to consider what happens inside a given interval $[\tau_\mu,\tau_{\mu+1}]$. 
Given a solution to~(\ref{cprob}), choose  $w_i^q[\mu]$ such that
\begin{equation}
(\tau_{\mu+1}-\tau_\mu)w_i^q[\mu]=\int_{\tau_{\mu}}^{\tau_{\mu+1}}\sum_k a_{ik}^q(t)dt, \ \forall i,q,\mu. \label{eq:equiv}
\end{equation}
This  ensures (\ref{lp1})--(\ref{lp3}) are satisfied with $\xi_i[\mu]=x_i(\tau_\mu)$, $\forall i,\mu$. It  follows from~\eqref{he4} and~\eqref{he5} that~\eqref{lp4} and~\eqref{lp5} are satisfied, respectively. One can similarly verify~\eqref{lp6} holds.  

Given a solution to~(\ref{probLP}) and the output from Algorithm~\ref{mc} for all intervals. It follows from the properties of McNaughton's algorithm~\cite{mc1959} that only one task is assigned to a processor at a time if $a$ is chosen to be piecewise constant such that
$
a_{ik}^q(t) = 1$ when $\sigma_{ik}^q[\mu](\tau_{\mu+1}-\tau_\mu)\leq t -\tau_\mu<\eta_{ik}^q[\mu](\tau_{\mu+1}-\tau_\mu)
$ and $
a_{ik}^q(t) = 0$ otherwise, $\forall i,k,q$.  After verifying that~\eqref{eq:equiv} holds, one can show that~\eqref{he1}--\eqref{he7} are satisfied.

The result follows by noting that the  costs of the two problems are equal with the above choices.
% , since
%\begin{equation*}
%\int_{\tau_\mu}^{\tau_{\mu+1}}  \sum_{i,k,q} a_{ik}^q(t)P(s^q)dt =\sum_qP(s^q)\int_{\tau_\mu}^{\tau_{\mu+1}}  \sum_{i,k} a_{ik}^q(t)dt
%%=\sum_qP(s^q)\int_{\tau_\mu}^{\tau_{\mu+1}}  \sum_{i,k} a_{ik}^q(t)dt
%\end{equation*}
\end{proof}

%\begin{theorem}\label{thm1} A feasible point $(a,x)$ satisfying~(\ref{he1})-(\ref{he7}) can be constructed from a solution to problem~(\ref{probLP}) + McNaughton's wrap around algorithm.
%\end{theorem}
%
%\begin{proof} Let $w_{ik}^q[\mu]$ be the output obtained from McNaughton's wrap around algorithm at time interval $[\tau_\mu,\tau_{\mu+1}]$. Then we can assign $a_{ik}^q(t) = 1$ if $w_{ik}^q[\mu] > 0$,~otherwise, $a_{ik}^q(t) = 0,~\forall i,q,t\in[\tau_\mu,\tau_{\mu+1}]$, hence (\ref{he4})-(\ref{he7}) are satisfied. The constraints~(\ref{he1})-(\ref{he3}) can be satisfied by noting that $x_i$ is a function of $a_{ik}^q$.
%\end{proof}
%\todo{Needs to be completed and more detailed}
%
%\begin{theorem} An optimal point to~(\ref{probLP}) can be used to construct an optimal point to~(\ref{cprob}).
%\end{theorem}
%\begin{proof}
%The existence of a feasible set is proved in Theorem~\ref{thm1}. Since the stage cost is linear in $a$, the value of the cost function of~(\ref{cprob}) is equal to that of~(\ref{probLP}). 
%\end{proof}
%\todo{Proof  not complete}

\subsection{Feedback Scheduler} \label{sec:feedback}
As can be seen, our open-loop optimal control problem is based on the estimated minimum execution time $\underline{x}_i$. The task will often finish earlier than expected, i.e.\ the actual minimum execution time $\underline{y}_i$ is often less than $\underline{x}_i$. Consider Figure~\ref{fig:fluidU}, which illustrates a fluid path of executing a task $T_i$. 
\begin{figure}[t]
\centering 
\includegraphics[width=\columnwidth]{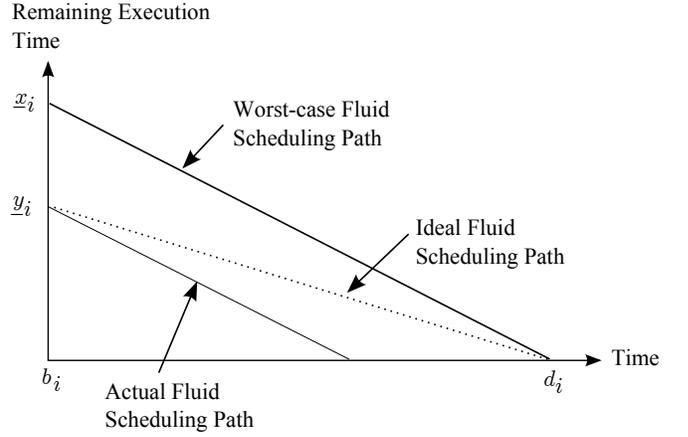}
\caption{Fluid scheduling model with uncertainty in the task execution time}
\label{fig:fluidU}
\end{figure}
Our open-loop algorithm follows a different path from the one that we really want to follow, i.e.\ the dotted line, due to uncertainty in task execution times. In other words, the open-loop algorithm can provide a solution that is overestimating the system workload, leading to higher energy consumption, due to the fact that the system operates at an unnecessarily higher speed. Therefore, it is better to feed back information whenever (i) a task finishes or (ii) a new task arrives at the system, in order to recalculate a new control action to respond to the changing workload. 

\begin{table*}[t]
\caption{Commercial processor details for simulation} \label{tab:pmsim}
\centering
\begin{tabular}{|l||c|c|c|c|c||c|c|c|c|}\hline
Processor type & \multicolumn{5}{c||}{XScale~\cite{xscale}} & \multicolumn{4}{|c|}{PowerPC 405LP~\cite{rusu04}} \\ \hline
Frequency (MHz) & 150 & 400 & 600 & 800 & 1000 & 33 & 100 & 266 & 333 \\
Speed	& 0.15 & 0.4 & 0.6 & 0.8 & 1.0 & 0.1 & 0.3 & 0.8 & 1.0 \\ 
Voltage (V) & 0.75 & 1.0 & 1.3 & 1.6 & 1.8 & 1.0 & 1.0 & 1.8 & 1.9 \\
Active Power (mW) & 80 & 170 & 400 & 900 & 1600 & 19 & 72 & 600 & 750 \\ \hline
Idle Power (mW) & \multicolumn{5}{c||}{40~\cite{Xu04}} & \multicolumn{4}{|c|}{12} \\\hline
\end{tabular}
\end{table*}

The overall architecture of our feedback scheduling system is given in Figure~\ref{fig:FBarch}, where the scheduler is called at two scheduling events. 
\begin{figure}[t]
\centering 
\includegraphics[width=\columnwidth]{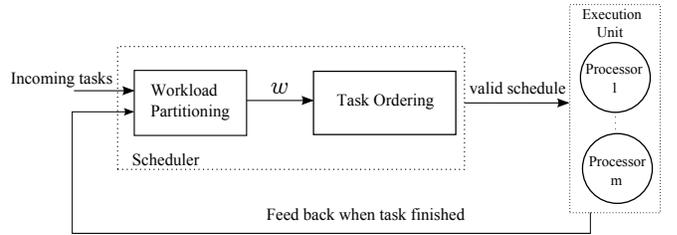}
\caption{Feedback scheduling architecture}
\label{fig:FBarch}
\end{figure}
One occurs when a task finishes its required executing workload/cycles on one of the processors and the other when a new task arrives. The scheduler is composed of two sub-units, i.e.\ a workload partitioning unit and a task ordering unit. By solving~(\ref{probLP}), the workload partitioning unit provides control input $w$ to the task ordering unit, which then uses McNaughton's wrap around algorithm to produce a valid schedule to the execution unit.

\begin{table}[t]
\caption{Simulation tasksets} \label{tab:taskset}
\centering
\begin{tabular}{|c|c|c|c|}\hline
$D$& $T_1$ & $T_2$ & $T_3$ \\ \hline
0.50 & (0,1,5) & (0,2,10) & (0,1.5,15)  \\ \hline
0.75 & (0,1,5) & (0,3.5,10) & (0,3,15) \\ \hline
1.00 & (0,2,5) & (0,4,10) & (0,3,15) \\ \hline
1.25 & (0,1,5) & (0,6.5,10) & (0,6,15)  \\ \hline
1.50 & (0,2,5) & (0,7,10) & (0,6,15)  \\ \hline
1.75 & (0,3,5) & (0,7.5,10) & (0,6,15) \\ \hline
2.00 & (0,4,5) & (0,6,10) & (0,9,15) \\ \hline
\multicolumn{4}{l}{Note: The second parameter of a task is $\underline{x}_i$; }\\
\multicolumn{4}{l}{$c_i$ can be obtained by multiplying $\underline{x}_i$ by $f_{max}$.} 
\end{tabular}
\end{table}

\section{Simulation and Results}\label{sec:sim}
To evaluate the performance of our feedback scheme, we consider a set of aperiodic tasks to be scheduled on two  commercial processors, namely a PowerPC 405LP and an XScale. The details of the two  processors are given in Table~\ref{tab:pmsim}. Two homogeneous systems composed of two processors of the same type were chosen. The energy consumed by executing each taskset, listed in Table~\ref{tab:taskset}, were evaluated. 

The minimum taskset density $D:=\sum_{i\in I}\delta_i$, a measurement of the utilization of computing resources in a given time interval, is defined as the sum of minimum task densities of all tasks within the system. The LP~\eqref{probLP} was modelled using OPTI TOOLBOX~\cite{opti2012} and solved with SoPlex~\cite{Wunderling1996}. %, which is included in~SCIP~\cite{scip2009}. 
%\todo{Check that this is the actual LP solver used - SCIP is an INTEGER solver, which we don't need - we only need the LP solver part.}

For this simulation, we only consider the scheduling event when a task finishes. %In other words, we assume that there are no new tasks arriving at the system during the simulation. 
Three algorithms are compared: (i)~Feedback LP-DVFS, which is our LP-DVFS + McNaughton's wrap around algorithm proposed in Section~\ref{sec:feedback}, (ii)~Open-Loop LP-DVFS, which is our LP-DVFS without feedback information on finishing tasks, and (iii)~No mismatch/Ideal, which is our LP-DVFS with the actual minimum task execution times equal to the estimated, i.e.\ $\underline{x}_i=\underline{y}_i$. %Though McNaughton's wrap around algorithm does not require the task execution order, for simplicity, our implementation assigns execution priority based on a task's deadline, i.e.\ earliest deadline will be executed first. 

Figure~\ref{fig:rPlot1} shows results from executing the tasksets in Table~\ref{tab:taskset} onto two homogeneous multiprocessor system, composed of two of each processor type, with the estimation factor $\gamma_i=0.5,~\forall i\in I$. 
\begin{figure*}[t]
\centerline{\subfloat[PowerPC 405LP]{\includegraphics[width=\columnwidth]{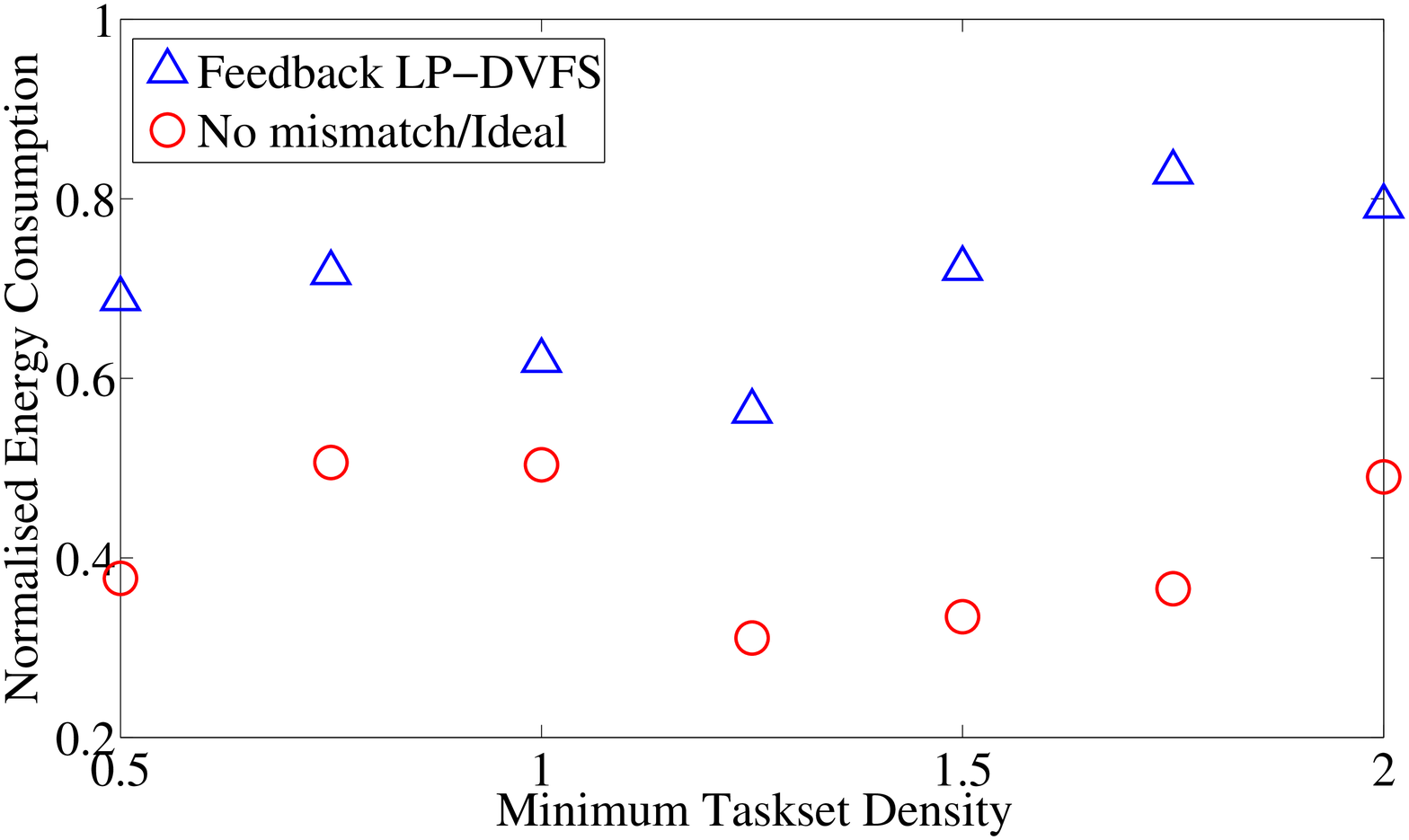}}
\subfloat[XScale]{\includegraphics[width=\columnwidth]{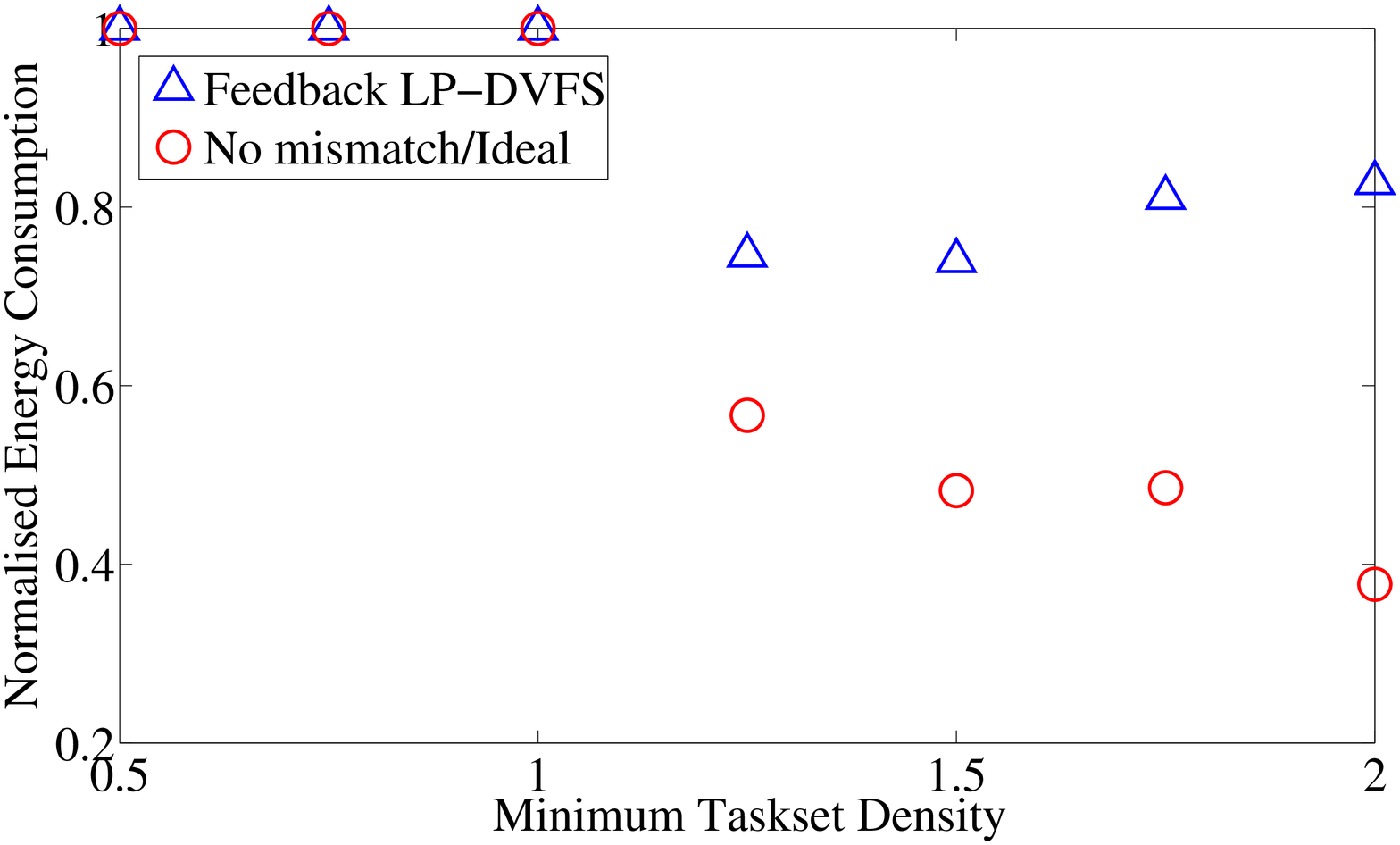}}}
\caption{Simulation results for different minimum taskset density with $\gamma_i=0.5,~\forall i\in I$.}
\label{fig:rPlot1}
\end{figure*}
The vertical axis is the total energy consumption normalised by the Open-Loop LP-DVFS algorithm. For a system composed of PowerPCs, the feedback scheme can save energy up to about 40\% compared to an open-loop scheme. However, for a system with XScale processors, the feedback scheme starts to perform better than the open-loop scheme only when the density is more than 1. Moreover, the  percentage saving of the XScale system is less than that of the PowerPC's. This is due to the differences in the distribution of speed levels of the two processor types, i.e.\ the XScale processor has more evenly distributed speed levels than that of the PowerPC; therefore, the optimizer can select the operating speed level that is closer to the optimal continous speed value.

The results from varying the estimation factor of the taskset with $D=1.25$ are shown in Figure~\ref{fig:rPlot2}. 
\begin{figure*}[t]
\centerline{\subfloat[PowerPC 405LP]{\includegraphics[width=\columnwidth]{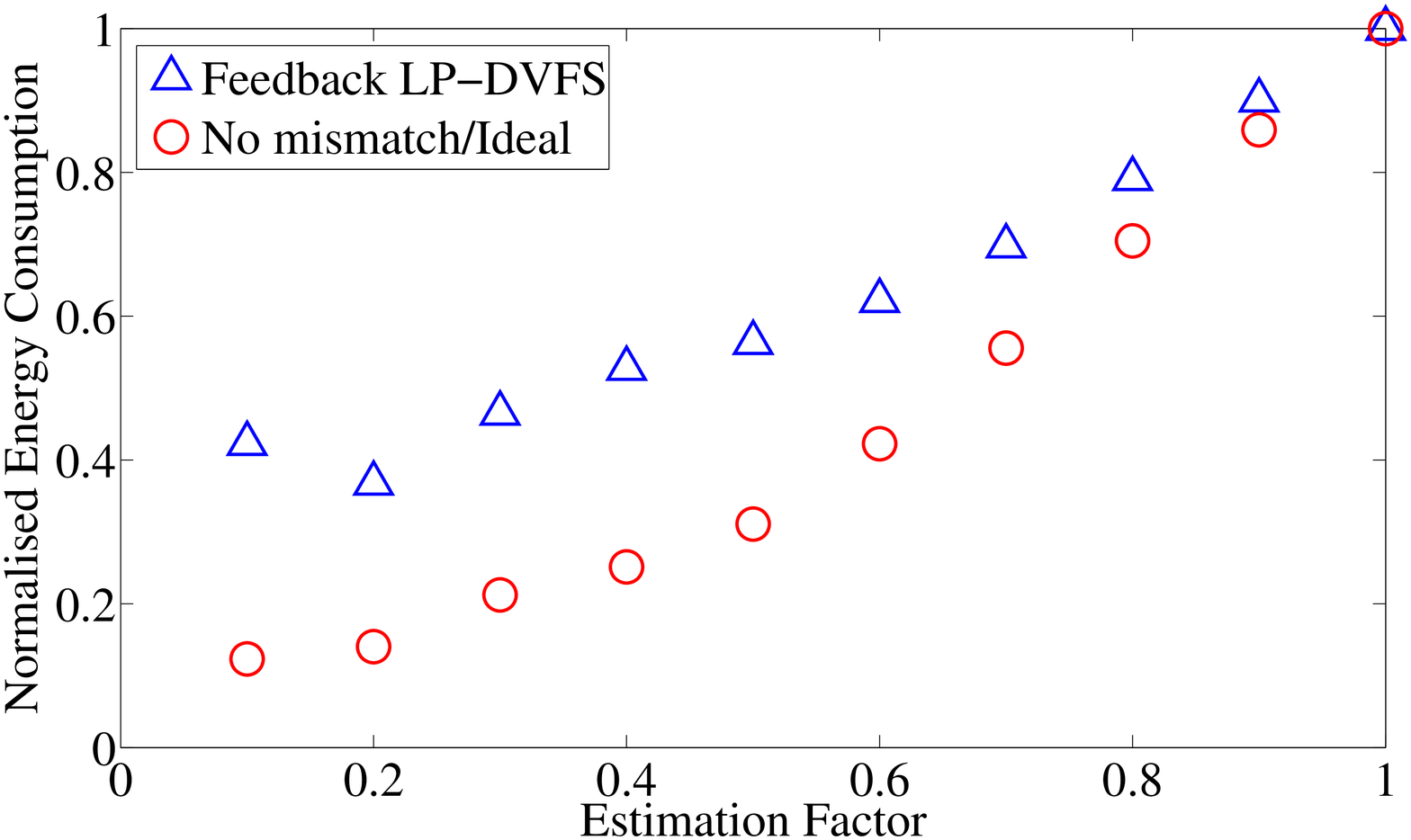}}
\subfloat[XScale]{\includegraphics[width=\columnwidth]{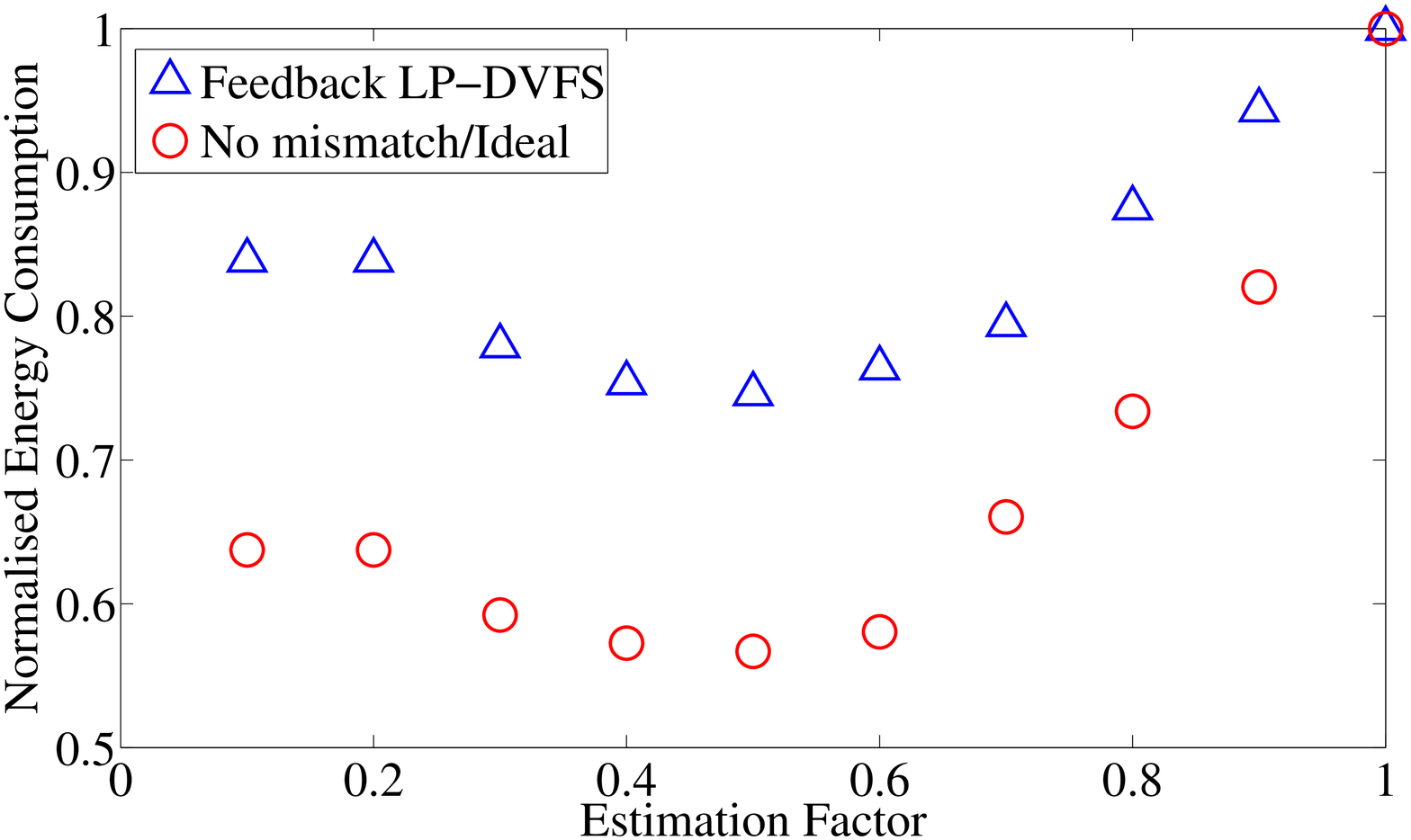}}}
\caption{Simulation results for different estimation factor $\gamma$ with $D=1.25$.}
\label{fig:rPlot2}
\end{figure*}
Note that, for this simulation, the estimation factors of all tasks are the same. For a PowerPC system, the energy saving is high when the estimation factor is low. In addition, the difference between the energy consumed by the feedback strategy and the ideal decreases as the estimation factor increases. On the other hand, for an XScale system, the maximum energy saving does not occur when the estimation factor is the lowest, but rather occurs at $\gamma = 0.5$. Furthermore, the energy consumption difference between the feedback and optimal/ideal is larger than that of the PowerPC's. Note that the energy saving varies with the tasksets, solutions from different LP solvers, and the task execution order.  Particularly, since the solutions are not unique, the choice of selecting the task execution order has an effect on the total energy consumption. %In other words, the future system behaviours are based on the present inputs.    

\section{Conclusions and Future Work}\label{sec:con}
A feedback method was adopted to solve a multiprocessor scheduling problem with uncertainty in task execution times. We have shown that our proposed closed-loop optimal control scheduling algorithm performs better than the open-loop algorithm in terms of energy efficiency. Simulation results suggest that the difference between closed-loop and open-loop performance can be reduced by having a more refined distribution of operating speed levels. 

The work presented here can be extended in a number of ways. For a periodic task, an estimator could be incorporated to obtain a better performance. %Specifically, one could dynamically compute the estimation factor to change the estimated execution time of each instance of the task. 
For further energy savings, a dynamic power management scheme (DPM), which determines when and how long the processor should be in the active or idle state, could also be integrated in the scheme.   

Finally, note that there are many links here to model predictive control~\cite{Mayne20142967} and it would therefore be of interest to investigate how methods developed in that community could be applied to the scheduling problem defined here. For example, one could extend the work to the problem of optimizing over feedback policies, rather than open-loop input sequences, as was done here. Efficient numerical methods, including distributed  cooperative schemes, could also be developed to solve the LP~\eqref{probLP} in real-time.
\balance
\bibliographystyle{IEEEtran}
\bibliography{FS}
\end{document}